\documentclass[letterpaper, 11pt]{article}

\usepackage{latexsym,amssymb,amsmath,amsthm,fullpage}

\newcommand{\mnote}[2]
{\mbox{}\marginpar%
[{\tiny\it\begin{minipage}[t]{\marginparwidth}\raggedleft#1%
\end{minipage}}]%
{\tiny\it\begin{minipage}[t]{\marginparwidth}\raggedright#1%
\end{minipage}}%
{}#2}

\newtheorem{theorem}{Theorem}[section]
\newtheorem{lemma}[theorem]{Lemma}

\newtheorem{corollary}[theorem]{Corollary}

\newtheorem{claim}[theorem]{Claim}

\theoremstyle{remark}

\theoremstyle{definition}
\newtheorem{definition}[theorem]{Definition}

\newcommand{\Real}{\mathbb R}
\newcommand{\expec}{{\mathbb{E}}}
\newcommand {\roundup}[1] {{\lceil {#1} \rceil}}
\newcommand {\rounddown}[1] {{\lfloor {#1} \rfloor}}
\newcommand{\eps}{\varepsilon}
\newcommand{\triv}{\gamma}
\newcommand{\prob}{{\rm Pr}}
\newcommand{\frc}{{\rm fr}}
\newcommand{\dks}{DkS}

\newtheorem*{Gnp-claim}{Claim \ref{G(n,p)-bound}}
\newtheorem*{dkslocal-claim}{Claim \ref{clm:dkslocal}}
\newtheorem*{expand-lemma}{Lemma \ref{lem:expand-guarantee}}
\newtheorem*{contract-lemma}{Lemma \ref{lem:contract-guarantee}}
\newtheorem*{main-cor}{Corollary \ref{cor:main}}

\begin{document}

\begin{titlepage}
\title{Detecting High Log-Densities -- an $O(n^{1/4})$ Approximation
  for Densest $k$-Subgraph}

\author{Aditya Bhaskara \thanks{Department of Computer Science, Princeton University, supported by NSF awards MSPA-MCS 0528414, CCF 0832797, and AF 0916218. Email: \textsf{bhaskara@cs.princeton.edu}} \and
Moses Charikar \thanks{Department of Computer Science, Princeton
University, supported by NSF awards MSPA-MCS 0528414, CCF 0832797, and AF 0916218. Email:
\textsf{moses@cs.princeton.edu}} \and Eden Chlamtac
\thanks{Weizmann Institute of Science, Rehovot, Israel, supported
by a Sir Charles Clore Postdoctoral Fellowship. Email:
\textsf{eden.chlamtac@weizmann.ac.il}}\and Uriel Feige
\thanks{Weizmann Institute of Science, Rehovot, Israel. Email:
\textsf{uriel.feige@weizmann.ac.il}. The author holds the Lawrence
G. Horowitz Professorial Chair at the Weizmann Institute. Work
supported in part by The Israel Science Foundation (grant No.
873/08).}\and Aravindan Vijayaraghavan
\thanks{Department of Computer Science, Princeton University,
supported by NSF awards MSPA-MCS 0528414, CCF 0832797, and AF 0916218. Email:
\textsf{aravindv@cs.princeton.edu}}}
\date{}

\maketitle

\begin{abstract}
In the Densest $k$-Subgraph problem, given a graph $G$ and a parameter
$k$, one needs to find a subgraph of $G$ induced on $k$ vertices that
contains the largest number of edges.  There is a significant gap
between the best known upper and lower bounds for this problem.  It is
NP-hard, and does not have a PTAS unless NP has subexponential time
algorithms.  On the other hand, the current best known algorithm of
Feige, Kortsarz and Peleg~\cite{FKP}, gives an approximation ratio of
$n^{1/3 - \eps}$ for some specific $\eps > 0$ (estimated by those
authors at around $\eps = 1/60$).

We present an algorithm that for every $\eps> 0$ approximates the
Densest $k$-Subgraph problem within a ratio of $n^{1/4 + \eps}$ in
time $n^{O(1/\eps)}$. If allowed to run for time $n^{O(\log n)}$, our
algorithm achieves an approximation ratio of $O(n^{1/4})$.  Our
algorithm is inspired by studying an average-case version of the
problem where the goal is to distinguish random graphs from random
graphs with planted dense subgraphs -- the approximation ratio we
achieve for the general case matches the ``distinguishing ratio'' we
obtain for this planted problem. Achieving a distinguishing ratio of
$o(n^{1/4})$ for the planted problem (in polynomial time) is beyond
the reach of our current techniques.

At a high level, our algorithms involve cleverly counting
appropriately defined trees of constant size in $G$, and using these
counts to identify the vertices of the dense subgraph.  Our algorithm
is based on the following principle. We say that a graph $G(V,E)$ has
log-density $\alpha$ if its average degree is
$\Theta(|V|^{\alpha})$. The algorithmic core of our result is a family
of algorithms that output $k$-subgraphs of nontrivial density whenever
the log-density of the densest $k$-subgraph is larger than the
log-density of the host graph.

Finally, we extend this algorithm to obtain an $O(n^{1/4 -
  \eps})$-approximation algorithm which runs in time
$O(2^{n^{O(\eps)}})$ and also explore various approaches
to obtain better approximation algorithms in restricted parameter
settings for random instances.
\end{abstract}

\end{titlepage}

\section{Introduction}

In this paper, we study the Densest $k$-Subgraph (DkS) problem: Given
a graph $G$ and parameter $k$, find a subgraph of $G$ on $k$ vertices
with maximum density (average degree). This problem may be seen as an
optimization version of the classical NP-complete decision problem
CLIQUE. The approximability of DkS is an important open problem and
despite much work, there remains a significant gap between the
currently best known upper and lower bounds.

In addition to being NP-hard (as seen by the connection to CLIQUE),
the DkS problem has also been shown not to admit a PTAS under various
complexity theoretic assumptions. Feige~\cite{F02} has shown this
assuming random 3-SAT formulas are hard to refute, while more recently
this was shown by Khot \cite{Khot} assuming that NP does not have
randomized algorithms that run in sub-exponential time (i.e.\ that $NP
\not \subseteq \cap_{\eps > 0} BPTIME(2^{n^{\eps}})$).

The current best approximation ratio of $n^{1/3-\eps}$ for some small
$\eps > 0$ (which has been estimated to be roughly $1/60$) %\footnote{The authors do not specify $\eps$ explicitly in the paper.
%In personal communication, they believe that $\eps$ is
%of the order of $1/60$}.
was achieved by Feige, Kortsarz and Peleg~\cite{FKP}. Other known
approximation algorithms have approximation guarantees that depend
on the parameter $k$. The greedy heuristic of Asahiro et
al.~\cite{Asahiro} obtains an $O(n/k)$ approximation. Linear and
semidefinite programming (SDP) relaxations were studied by
Srivastav and Wolf~\cite{SW} and by Feige and
Langberg~\cite{FL01}, where the latter authors show how they can
be used to get approximation ratios somewhat better than $n/k$.
Feige and Seltser~\cite{FS} show graphs for which the
integrality gap of the natural SDP relaxation is
$\Omega(n^{1/3})$, indicating that in the worst case, the
approximation ratios achieved in~\cite{FKP} are better than those
achievable by this SDP. When the input is a complete graph with
edge weights that satisfy the triangle inequality, a simple greedy
algorithm achieves an approximation ratio of~2 (though the
analysis of this algorithm is apparently not easy, see~\cite{BG}).

A related problem to \dks\ is the max density subgraph problem,
where the aim is to find a subgraph $H$ which maximizes the ratio
of number of edges to number of vertices in $H$. It turns out that
this can be solved in polynomial time~\cite{GGT}. Charikar et
al.~\cite{CHK} recently showed an $O(n^{1/3})$ approximation to
the maximization version of label cover. This problem is at least
as difficult as \dks\ in the sense that there is an approximation
preserving randomized reduction from \dks\ (see~\cite{CHK} for example) to it. No reduction in
the opposite direction
%(which would hint that \dks~ is hard)
is known.

Our algorithm for \dks\ is inspired by studying an average-case
version we call the `Dense vs Random' question (see Section
\ref{sec:random} for a precise definition).  Here the aim is to
distinguish random graphs from graphs containing a dense
subgraphs, which can be viewed as the task of efficiently certifying
that random graphs do not contain dense subgraphs. This distinguishing
problem is similar in flavour to the well-studied planted clique
problem %,where the aim is to find a clique of size $k$ in $G(n,1/2)$
(see \cite{AKS}). Getting a better understanding of this planted
question seems crucial for further progress on \dks.
%seems to be a barrier for developing improved approximation algorithms for \dks.

%An additional motivation for studying the $DkS$ problem and its random planted variant comes from the fact that it
Some recent papers have used the hypothesis that (bipartite versions of) the planted dense subgraph problem is computationally
hard:
%The conjectured hardness of planted versions of \dks~ are sometimes used to prove conditional hardness results for other problems.
Applebaum et al.~\cite{ABW} use this in the design of a new public key encryption scheme.
%one of whose assumptions, is that (roughly speaking) a planted
%variant of densest subgraph in unbalanced bipartite graphs is hard.
More recently,  Arora et al.~\cite{ABBG} use this to demonstrate that evaluating certain financial derivatives is computationally hard.
%have used the conjectured hardness of another planted variant of \dks~(namely, finding a small dense biparitite graph planted in a bipartite random graph), to show that it is possible to construct toxic financial derivative products, such that detecting the bad assets and that a fraud has been committed are compuationally intractable.
The use of such hardness assumptions provides additional motivation for the study of algorithms for these problems.

\subsection{Our results} \label{subsec:results}

Our main result is a polynomial time $O(n^{1/4 + \eps})$ approximation algorithm for \dks, for any constant $\eps > 0$. That is, given $\eps >0$, and a graph $G$ with a $k$-subgraph of density $d$, our algorithm outputs a $k$-subgraph of density $\Omega \left(d/{n^{1/4+\eps}} \right)$ in polynomial time. In particular, our techniques give an $O(n^{1/4})$-approximation algorithm running in $O(n^{\log n})$ time.

At a high level, our algorithms involve cleverly counting appropriately defined subgraphs of constant size in $G$, and use these counts to identify the vertices of the dense subgraph. A key notion which comes up in the analysis is the following:
\begin{definition}
\label{def:logdensity} The {\em log-density} of a graph $G(V,E)$
with average degree $D$ is $\log_{|V|} D$. In other words, if a
graph has log-density $\alpha$, its average degree is $|V|^{\alpha}$.
\footnote{We will ignore low order terms when
expressing the log-density. For example, graphs with constant
average degree will be said to have log-density~0, and cliques
will be said to have log-density~1.}
\end{definition}

We first consider the random setting -- distinguishing between $G$ drawn from $G(n,p)$, and $G$ containing a $k$-subgraph $H$ of certain density planted in it. In fact, we examine a few variants (i.e.\ in the second case each of $G$ and $H$ may or may not be random). For all these variants we show that if the log-density of $G$ is $\alpha$ and that of $H$ is $\beta$, with $\beta > \alpha$, we can solve the distinguishing problem in time $n^{O(1/(\beta - \alpha))}$.

Our main technical contribution is that a result of this nature can be proven for arbitrary graphs. Informally, our main result, which gives a {\em family} of algorithms, parametrized by a rational number $r/s$, can be stated as follows (see Theorem~\ref{thm:main} for a more precise statement):
\begin{theorem}\label{thm:mainintroduction} {\bf (informal)}
Let $s>r>0$ be relatively prime integers, let $G$ be an undirected graph with maximum degree $D=n^{r/s}$, which
contains a $k$-subgraph $H$ with average degree $d$.
Then there is an algorithm running in time $n^{O(r)}$ that finds
a $k$-subgraph of average degree $\Omega(d/D^{(s-r)/s})$.
\end{theorem}

Note that the log-density of $H$ does not explicitly occur in the statement of the theorem. However, it turns out we can pre-process the graph, and restrict ourselves to the case $kD = n$ (see Appendix~\ref{sec:boundproduct}), in which case $D^{(s-r)/s} = k^{r/s}$, thus the output subgraph has average degree $d/k^{r/s}$. So if the log-density of $H$ is $\beta$ (recall that $G$ has log-density $\leq r/s$), the output graph has density $d/k^{r/s} = k^{\beta - r/s}$. Thus the difference in the log-densities also plays a role in the case of arbitrary graphs.

Also note that the theorem deals with the \emph{maximum} degree in $G$, and not average degree (which defines the log-density). It turns out that this upper-bound on the log-density will suffice (and will be more useful in the analysis).

As we observed earlier, we give a family of algorithms
parameterized by a rational number. Thus, given $G$ and $k$, we
pick $r/s$ appropriately and appeal to the theorem. In some sense,
this family of algorithms is a systematic generalization of the
(somewhat ad hoc) algorithms of \cite{FKP}.

Finally, observe that theorem implies an approximation ratio of at most
$D^{(s-r)/s} \leq n^{r(s-r)/s^2} \le n^{1/4}$ for every choice of $s>r>0$. As we mentioned, the statement above is informal. If we choose to restrict the running time to $O(n^{s_0})$ by limiting ourselves to $r<s\leq s_0$ (i.e.\ the bound on $D$ will not be exact), we lose a factor $n^{1/s_0}$ in the approximation. We refer to Section~\ref{sec:main} for the details.

\paragraph*{Outline of techniques.} The distinguishing algorithm for the Dense vs Random problem is based on the fact that in $G(n,p)$, instances of any fixed constant size structure appear more often when the graph has a higher log-density. More precisely, given parameters $r,s$, we will define a (constant size) tree $T_{r,s}$ such that a fixed set of leaves can be completed to many instances of $T_{r,s}$ in a graph with log-density $>r/s$, whereas in a random graph with log-density $<r/s$ there will only be a negligible number of such instances. Thus, if the log-density of $H$ is greater than $r/s$, finding a small set of vertices in $H$ (and using them as leaves of $T_{r,s}$) can help reveal larger portions of a dense subgraph. Though our intuition comes from random graphs, the heart of the argument carries over to worst-case instances.
%, since bounds on the number of instances of a tree depend only on the log-density of a graph (this is not true for other small subgraphs, for instance triangles or four-cycles, which appear in random graphs of any log-density, but can be absent from specific graphs with relatively high log-density).

We use a linear programming relaxation to guide us in our search for
the fixed vertices' assignment and obtain the dense subgraph. In order
to extract consistent feasible solutions from the LP even under the
assumptions that fixed vertices belong to $H$, the LP relaxation will
have a recursive structure similar to the Lov\'asz-Schrijver hierarchy
\cite{LS}. Feasible solutions to this LP (when they exist) can be
found in time $n^{O(r)}$ (where the depth of the recursion will be
roughly $r$), while the rest of the algorithm (given the LP solution)
will take linear time. As we shall see, there is also a combinatorial
variant of our algorithm, which, rather than relying on an LP
solution, finds the appropriate set of leaves by exhaustive search (in
time $n^{r+O(1)}$). While the analysis is essentially the same as for
the LP variant, it is our hope that a mathematical programming
approach will lead to further improvements in running time and
approximation guarantee.

The approximation ratio we achieve for general instances of DkS
matches the ``distinguishing ratio" we are currently able to achieve for various
random settings. This suggests the following concrete open
problem which seems to be a barrier for obtaining an approximation
ratio of $n^{1/4-\eps}$ for \dks\ -- distinguish between the
following two distributions:
\begin{quote}
$\mathcal{D}_1$: graph G picked from $G(n,n^{-1/2})$, and\\
$\mathcal{D}_2$: graph G picked from $G(n,n^{-1/2})$ with the induced
  subgraph on $\sqrt{n}$ vertices replaced with $G(\sqrt{n},n^{-(1/4+\eps)})$.
\end{quote}
In section~\ref{sec:subexponential} we will see that this
distinguishing problem can be solved in time $2^{n^{O(\eps)}}$, and
this can be used to give an algorithm for DkS with approximation ratio
$n^{1/4 - \eps}$, and run time $2^{n^{O(\eps)}}$. These mildly
exponential algorithms are interesting given the recent results
of \cite{ABBG} and \cite{ABW}, which are based on the assumption that
planted versions of DkS are hard.

In section~\ref{sec:sublogdensity}, we show that in the random setting we can beat the
log-density based algorithms for certain ranges of parameters.  We use different
techniques for different random models, some of which are very different from those used in sections~\ref{sec:random} and \ref{sec:main}. Interestingly, none of these techniques give a distinguishing ratio better than $n^{1/4}$ when $k=D=\sqrt{n}$.
% and involve looking at the second
%eigenvalue of the graph.
%
%
%counts the number of occurences of a particular template or structure
%(which depends on the parameter $p/q$), after fixing
%some of the vertices of the template. The algorithm proceeds by showing that there
%exists an assignment of these fixed vertices, such that
%we can find a dense subgraph between the candidate vertices of two adjacent vertices
%in the template.
% The algorithm in Section \ref{sec:random} deals with
%finding dense subgraphs when they are planted in sparse random graphs,
%and motivates the algorithm for general graphs. The algorithm for general graphs is
%covered in detail in Section \ref{sec:main}.

\subsection{Organization of paper}

In Section~\ref{sec:notation.and.simplification}, we introduce
some notation, and describe simplifying assumptions which will be
made in later sections (some of these were used in \cite{FKP}). In
Section~\ref{sec:random}, we consider two natural `planted'
versions of DkS, and present algorithms for these versions. The
analysis there motivates our approximation algorithm for DkS,
which will be presented in Section~\ref{sec:main}. In
Section~\ref{sec:subexponential} and Section~\ref{sec:sublogdensity},
we explore approaches to overcome the log-density barrier that
limits our algorithms in Section~\ref{sec:main}. In
Section~\ref{sec:subexponential} we give an $O(n^{1/4 -\eps})$
approximation algorithm for arbitrary graphs with run time
$O(2^{n^{O(\eps)}})$ time, and in Section~\ref{sec:sublogdensity},
we show that in various random %the Dense vs Random
settings, we can obtain a $\sqrt{D}$-approximation (which is better
than the log-density guarantee for $1<D<\sqrt{n}$).
% using eigenvalue techniques.

\iffalse
we introduce some notation and then present some simplifying assumptions, which are justified by various preprocessing and postprocessing steps (some of which were used in \cite{FKP}). % that we can use in our
%algorithms, so that the presentation of our algorithms can be
%simplified.
In Section~\ref{sec:random} we motivate our algorithm by first examining the case of %for detecting log-densities, we first discuss the densest $k$-subgraph
%problem on random instances Section~\ref{sec:random} (planted variants). We initially consider the simple case where a
a random dense $k$-subgraph is planted in another random graph on $n$ vertices %. While this problem is solved easily and
(which does not seem to %have potential to
generalize well to arbitrary (worst case) instances), and then consider the more general framework of  distinguishing between a random graph and an arbitrary graph containing a dense $k$-subgraph, which will capture the essence of our approach. %
%a arbitrary dense $k$-subgraph planted in a random graph on $n$ vertices seems to capture some of the difficulty in the problem.
%We come up with an algorithm that detects the presence of higher log-densities in this case of an arbitrary dense $k$-subgraph in a random graph (it performs this detection in astronger sense as seen in section ~\ref{sec:random}.
Finally, in Section~\ref{sec:main}, we describe our main algorithm, which achieves an $O(n^{1/4+\eps})$ approximation in general graphs.
\fi

\section{Notation and some simplifications}
\label{sec:notation.and.simplification}
We now introduce some notation which will be used in the rest of the paper. Unless otherwise stated,
$G(V,E)$ refers to an input graph on $n$ vertices, and $k$ refers to the size of the subgraph we are required to output. Also, %Let $D$ be the average degree of $G$. Let
$H$ will denote the densest $k$-subgraph (breaking ties arbitrarily) in $G$, and $d$ denotes the average degree of $H$. %equally dense subgraphs, $H$ is one of them chosen arbitrarily and let $d$ be its average degree.
For $v \in V$, $\Gamma(v)$ denotes the set of neighbors of $v$, and %$\Gamma^i(v)$ for $i \ge 2$ - set of vertices at distance $i$ from $v$. $\Gamma_H(v)$, $\Gamma^i_H(v)$ for $v \in H$ and $i \ge 2$ are correspondingly defined wrt to subgraph $H$. Let $\alpha$ and $\beta$ be the log-densities of $G$ and $H$ respectively.
for a set of vertices $S\subseteq V$, $\Gamma(S)$ denotes the set of all neighbors of vertices in $S$. Finally, for any number $x\in\Real$, will use the notation $\frc(x)=x-\rounddown{x}$.
%Finally $(r,s)$-caterpillar is a caterpillar with a total of $s$ edges and
%$r-1$ hairs (connected to carefully chosen backbone vertices).

We will make the following simplifying assumptions in the remaining sections: (these are justified in Section~\ref{sec:simplifications} of the appendix)
\begin{enumerate}
\item There exists a $D$ such that (a) the maximum degree of $G$ is at most $D$, and (b) a greedy algorithm finds a $k$-subgraph of density $\max\{1,kD/n\}$ in $G$. %
\item $d$ is the minimum degree in $H$ (rather than the average degree)
\item It suffices to find a subgraph of size {\em at most} $k$, rather than exactly $k$. In Section~\ref{sec:main} we use `$k$-subgraph' more loosely to mean a subgraph on at most $k$ vertices.
\item When convenient, we may also take $G$ (and hence $H$) to be bipartite.
\item The edges of the graph $G$ are assumed to be unweighted, since
  we can bucket the edges into $O(\log n)$ levels according to the
  edge weights (which we assume are all positive), and output the
  densest of the $k$-subgraphs obtained by applying the algorithm to
  each of the graphs induced by the edges in a bucket. This incurs a
  loss of just $O(\log n)$ factor in the approximation.
\end{enumerate}
In many places, we will ignore leading constant factors (for example, we may find a subgraph of size $2k$ instead of $k$). It will be clear that these do not seriously affect the approximation factor. %So, when convenient, we may also take $G$ (and hence $H$) to be bipartite.
%
%It also suffices to treat $O(k)$ as an upper bound (appendix \ref{sec:lek}). If the subgraph obtained has size $k'<k$, we can remove from $G$ the edges of the subgraph found and recurse. If $k'=ck$, we just lose constant factors to sample and obtain a subgraph of size $k$.
%As is known also from previous work, since we do not care about
%constant factors in the approximation ratio then $G$ may be
%assumed to be bipartite (appendix \ref{sec:bipartite}).
%It will also be convenient for us to view $H$ (the densest $k$-subgraph
%of $G$) as having minimum degree $d$ rather than average degree using standard arguments (appendix \ref{sec:ged}).
%Also, constant factors will be ignored throughout.

% Finally, it will be instructive to consider the case where
% $kD \le n$. This assumption can always be made, but at the cost of an extra $O(\sqrt{\log n})$ factor in the approximation guarantee, and of introducing randomness into an otherwise deterministic algorithm (appendix \ref{sec:boundproduct}). While we do not make this assumption in the proofs in Section~\ref{sec:random} and Section~\ref{sec:main}, the analysis would become slightly simpler, and it may be helpful to consider for the sake of intuition.

\section{Random graph models}
\label{sec:random}

An $f(n)$-approximation algorithm for the densest $k$-subgraph problem must be able to distinguish between graphs where any $k$-subgraph has density at most $c$, and graphs with an $cf(n)$-dense $k$-subgraph planted in them. Random graphs are a natural class of graphs that do not contain dense $k$-subgraphs. Further, random graphs seem to present challenges for currently best known algorithms for \dks. Hence, it is instructive to see what parameters allow us to efficiently solve this distinguishing problem.

We consider three variants of the random distinguishing problem, in increasing order of difficulty. In the \emph{Random Planted Model}, we would like to distinguish between two distributions:

\begin{quote}
$\mathcal{D}_1$: Graph $G$ is picked from $G(n,p)$, with $p=n^{\alpha-1}$, $0<\alpha<1$.\\
$\mathcal{D}_2$: $G$ is picked from $G(n, n^{\alpha-1})$ as before. A set $S$ of $k$ vertices is chosen arbitrarily, and the subgraph on $S$ is replaced with % all edges within $S$ are removed, and instead one puts
a random graph $H$ from $G(k, k^{\beta-1})$ on $S$.%\footnote{We also allow removing edges from $S$ to $V \setminus S$, so that tests based on simply looking at the degrees do not work.}
\end{quote}

A slightly harder variant is the \emph{Dense in Random} problem, in which we would like to distinguish between $G$ chosen from $\mathcal{D}_1$, as before, and $G$ which is chosen similarly to $\mathcal{D}_2$, except that the planted subgraph $H$ is now an \emph{arbitrary} graph with average degree $k^\beta$ (that is, log-density $\beta$). Here, the algorithm must be able to detect the second case with high probability regardless of the choice of $H$.

Finally, we consider the \emph{Dense versus Random} problem, in which we would like to distinguish between $G \sim \mathcal{D}_1$, and an \emph{arbitrary} graph $G$ which contains a $k$-subgraph $H$ of log-density $\beta$.

Observe that for $G \sim \mathcal{D}_1$, a $k$-subgraph would have expected average degree $kp = kn^{\alpha-1}$. Further, it can be shown that densest $k$-subgraph in $G$ will have average degree $\max \{kn^{\alpha-1}, 1\}$, w.h.p. (up to a logarithmic factor). Thus if we can solve the distinguishing problem above, its `distinguishing ratio' would be $\min_{\beta}(k^{\beta} / \max \{kn^{\alpha-1}, 1\})$, where $\beta$ ranges over all values for which we can distinguish (for the corresponding values of $k,\alpha$). If this is the case for all $\beta>\alpha$, then (as follows from a straightforward calculation), the distinguishing ratio is never more than
\begin{align*}
\frac{k^{\alpha}}{\max \{kn^{\alpha-1}, 1\}} &=\min\left\{\left(\frac{n}k\right)^{1-\alpha},k^{\alpha}\right\}\\
&=n^{\alpha(1-\alpha)}\cdot\min\left\{\left(\frac{n^{1-\alpha}}k\right)^{1-\alpha},\left(\frac k{n^{1-\alpha}}\right)^{\alpha}\right\}\\
&\leq n^{\alpha(1-\alpha)}\\&\leq n^{1/4}.
\end{align*}

In this section we will only discuss the Random Planted Model and the Dense versus Random problem, while the intermediate Dense in Random problem is only examined in Section~\ref{sec:sublogdensity}.

\subsection{The random planted model}
\iffalse
The most natural setting is the problem of distinguishing between the two distributions
\begin{quote}
$\mathcal{D}_1$: Graph $G$ is picked from $G(n,p)$, with $p=n^{\theta-1}$, $0<\theta<1$.\\
$\mathcal{D}_2$: $G$ is picked from $G(n, n^{\theta-1})$ as before. A set $S$ of $k$ vertices is chosen arbitrarily, and the subgraph on $S$ is replaced with % all edges within $S$ are removed, and instead one puts
a random graph $H$ from $G(k, k^{\theta'-1})$ on $S$.%\footnote{We also allow removing edges from $S$ to $V \setminus S$, so that tests based on simply looking at the degrees do not work.}
\end{quote}

Observe that for $G \sim \mathcal{D}_1$, a $k$-subgraph would have expected average degree $kp = kn^{\theta-1}$. Further, it can be shown that densest $k$-subgraph in $G$ will have average degree $\max \{kn^{\theta-1}, 1\}$, w.h.p. (up to a logarithmic factor). Thus if we can solve the distinguishing problem above, its `distinguishing ratio' would be $\min_{\theta'}(k^{\theta'} / \max \{kn^{\theta-1}, 1\})$, where $\theta'$ ranges over all values for which we can distinguish (for the corresponding values of $k,\theta$). If this is the case for all $\theta'>\theta$, then (as follows from a straightforward calculation), the distinguishing ratio is never more than $$k^{\theta} / \max \{kn^{\theta-1}, 1\}\leq n^{\theta(1-\theta)}\leq n^{1/4}.$$
\fi

One easy way of distinguishing between the two distributions in the Random Planted Model involves looking at the highest degree vertices, or at the pairs of vertices with the largest intersection of neighborhoods. This approach, which is discussed in Section~\ref{sec:sublogdensity} is not only a distinguishing algorithm, but can also identify $H$ in the case of $G\sim\mathcal{D}_2$. However, it is not robust, in the sense that we can easily avoid a detectable contribution to the degrees of vertices of $H$ by resampling the edges between $H$ and $G\setminus H$ with the appropriate probability.

Rather, we examine a different approach, which is to look for constant size subgraphs $H'$ which act as `witnesses'. If $G \sim \mathcal{D}_1$, we want that w.h.p.\ $G$ will not have a subgraph isomorphic to $H'$, while if $G \sim \mathcal{D}_2$, w.h.p.\ $G$ should have such a subgraph. It turns out that whenever $\beta > \alpha$, such an $H'$ can be exists, and thus we can solve the distinguishing problem.

Standard probabilistic analysis (cf.~\cite{AS}) shows that if a graph has log-density
greater than $r/s$ (for fixed integers $0 < r < s$) then it is
expected to have constant size subgraphs in which the ratio of
edges to vertices is $s/(s-r)$, and if  the log-density is smaller
than $r/s$, such subgraphs are not likely to exist (i.e., the
occurrence of such subgraphs has a threshold behavior).
Hence such subgraphs can serve as witnesses when $\alpha<r/s<\beta$.
%to the existence of $H$ when the log-density of $H$ is larger than that of $G$.

Observe that in the approach outlined above, $r/s$ is rational,
and the size of the witnesses increases as $r$ and $s$ increase.
This serves as intuition as to why the statement of
Theorem~\ref{thm:mainintroduction} involves a rational number
$r/s$, with the running time depending on the value of $r$.

\subsection{Dense versus Random} \label{sec:dvsr}
The random planted model above, though interesting, does not seem to say much about the general \dks\ problem. In particular, for the Dense versus Random problem,
 %We consider an `intermediate' problem, which we call the Dense vs Random question. The aim is to distinguish between $\mathcal{D}_1$ exactly as above, and $\mathcal{D}_2$ similar to the above, except that the planted graph $H$ is an {\em arbitrary} graph of log-density $\theta'$  instead of a random graph. Due to this,
simply looking for the occurrence of subgraphs need not work, because the planted graph could be very dense and yet not have the subgraph we are looking for.

To overcome this problem, we will use a different kind of witness,
which will involve special constant-size trees, which we call {\em
  templates}. In a {\em template witness} based on a tree $T$, we fix
a small set of vertices $U$ in $G$, and count the number of trees
isomorphic to $T$ whose set of leaves is exactly $U$. The
templates are chosen such that a random graph with log-density
below a threshold will have a count at most poly-logarithmic for
{\em every} choice of $U$, while we will show by a counting
argument that in any graph (or subgraph) with log-density above
the same threshold, there exists a set of vertices $U$ which
coincide with the leaves of at least $n^{\eps}$ copies of $T$ (for
some constant $\eps>0$). As noted in
Section~\ref{sec:notation.and.simplification}, we may assume
\emph{minimum} degree $k^{\beta}$ in $H$ as opposed to average
degree (this will greatly simplify the counting argument).

As an example, suppose the log-density is $2/3$. In this case, the
template $T$ we consider is the tree $K_{1,3}$ (a claw with three
leaves). For any triple of vertices $U$, we count the number of copies
of $T$ with $U$ as the set of leaves -- in this case this is precisely
the number of common neighbors of the vertices in $U$. In this case,
we show that if $G \sim \mathcal{D}_1$, with $\alpha \leq 2/3$, {\em
  every} triple of vertices has at most $O(\log n)$ common
neighbors. While in the dense case, with $\beta = 2/3 +
\eps$, there exists some triple with at least $k^{\eps}$ common
neighbors. Since for ranges of parameters of interest $k^{\eps} =
\omega(\log n)$, we have a distinguishing algorithm.

Let us now consider a log-density threshold of $r/s$ (for some
relatively prime integers $s>r>0$). The tree $T$ we will associate
with the corresponding template witness will be a caterpillar -- a
single path called the {\em backbone} from which other paths, called
{\em hairs}, emerge. In our case, the hairs will all be of length
$1$. More formally,

\begin{definition}\label{rs-caterpillar} An $(r,s)${\em-caterpillar} is a tree constructed inductively as follows: Begin with a single vertex as the leftmost node in the backbone. For $s$ steps, do the following: at step $i$, if the interval $[(i-1)r/s, ir/s]$ contains an integer, add a hair of length $1$ to the rightmost vertex in the backbone; otherwise, add an edge to the backbone (increasing its length by $1$).
\end{definition}

This inductive definition is also useful in deriving an upper bound on
the number of $(r,s)$-caterpillars in $G(n,p)$ (for $p\leq n^{r/s-1}$)
with a fixed sequence of `leaves' (end-points of the hairs)
$v_0,v_1,\ldots,v_r$. We do this by bounding the number of candidates
for each internal (backbone) vertex, and showing that with high
probability, this is at most $O(\log n)$. We begin by bounding the
number of candidates for the rightmost backbone vertex in a {\em
  prefix} of the $(r,s)$ caterpillar (as per the above inductive
construction).  For each $t=1,\ldots,r$, let us write
$S_{v_0,\ldots,v_{\lfloor tr/s\rfloor}}(t)$ for the set of such
candidates at step $t$ (given the appropriate prefix of
leaves).
 \iffalse Further, we will ask that the candidate vertices for these
backbone vertices come from disjoint sets $V_0,V_1,\dots,V_{\lfloor
  tr/s \rfloor}$. This will ensure that the events that $u \in S(t-1)$
and $(u,v) \in E$ are independent. Further, since we are partitioning
into a constant number of sets, the counts we are interested in are
preserved up to a constant -- this is essentially the color-coding
trick of Alon, Yuster and Zwick~\cite{AYZ}.\footnote{Roughly, the
  color coding theorem says the following: suppose we randomly color
  the vertices of a graph $G$ with $C$ colors. Suppose $G$ has $M$
  copies of a template that has $C$ vertices. Then w.h.p, there exist
  $\frac{M}{C^C}$ `colorful' copies of $\mathrm{L}$ (each vertex of
  the template has a different color).}\fi
 The following claim
upper bounds the cardinality of these sets (with high
probability). (Recall the notation $\frc(x)=x-\rounddown{x}$.)

\begin{claim}\label{G(n,p)-bound} In $G(n,p)$, for $p\leq n^{r/s-1}$, for every $t=1,\ldots,s$ and for any fixed sequence of vertices $U_i=v_0,\ldots,v_{\lfloor tr/s\rfloor}$, for every vertex $v\in V\setminus U_i$ we have $$\prob[v\in S_{v_0,\ldots,v_{\lfloor tr/s\rfloor}}(t)]\leq n^{\frc(tr/s)-1}(1+o(1)).$$
\end{claim}

Intuitively, the claim follows from two simple observations: (a) For
any set of vertices $S\subseteq V$ in $G(n,p)$, w.h.p.\ the
neighborhood of $S$ has cardinality at most $pn|S|$ (since the degree
of every vertex is tightly concentrated around $pn$), and (b) for
every vertex set $S$, the expected cardinality of its intersection
with the neighborhood of any vertex $v$ is at most
$\expec[|S\cap\Gamma(v)|]\leq p|S|$. Applying these bounds inductively
to the construction of the sets $S(t)$ when $p=n^{r/s-1}$ then implies
$|S(t)|\leq n^{\frc(tr/s)}$ for every $t$.

\begin{proof}[Proof (sketch)] In fact, it suffices to show equality for $p=n^{r/s-1}$ (since for sparser random graphs the probability can only be smaller). More precisely, for this value of $p$, we show:
$$\prob[v\in S_{v_0,\ldots,v_{\lfloor tr/s\rfloor}}(t)]=n^{\frc(tr/s)-1}(1\pm o(1)).$$

We prove the claim by induction. For $i=1$, it follows by definition of $G(n,p)$: $\prob[v\in S_{v_0}(1)]=p=n^{r/s-1}$. For $t>1$, assume the claim holds for $t-1$. If the interval $[(t-1)r/s, tr/s]$ contains an integer (for $1<t\leq s$ it must be $\roundup{(t-1)r/s}$), then $S(t)=S(t-1)\cap\Gamma(v_{\roundup{(t-1)r/s}})$. Thus, by definition of $G(n,p)$ and the inductive hypothesis,
\begin{align*}\prob[v\in S_{v_0,\ldots,v_{\lfloor tr/s\rfloor}}(t)] &=p \cdot \prob[v\in S_{v_0,\ldots,v_{\lfloor (t-1)r/s\rfloor}}(t-1)] \\
&=
  n^{r/s-1}n^{\frc((t-1)r/s)-1}(1+o(1))\\ &=n^{\frc(tr/s)-1}(1\pm o(1)).
\end{align*}
Otherwise, if the interval $[(t-1)r/s, tr/s]$ does not contain an
integer, then $S(t)=\Gamma(S(t-1))$. In this case, by the inductive
hypothesis, the cardinality of the set $|S(t-1)|$ is tightly
concentrated around $n^{\frc((t-1)r/s)}$ (using Chernoff-Hoeffding
bounds). If we condition on the choice of all $S(t')$ for $t'<t$, and $S(t-1)$ has (approximately) the above cardinality, then for every $v$ not appearing in the previous sets, we have
\begin{align*}
\prob[v\in S_{v_0,\ldots,v_{\lfloor tr/s\rfloor}}(t)] &= \prob[\exists u\in S_{v_0,\ldots,v_{\lfloor tr/s\rfloor}}(t-1): (u,v)\in E]% \cap v \in V_{\lfloor tr/s \rfloor}] \\
\\
&=1-(1-p)^{|S(t-1)|}\\
&=p|S(t-1)|(1-o(1))&\text{since }p|S(t-1)|=o(1)\\
&=n^{r/s-1}n^{\frc((t-1)r/s)}(1\pm o(1))\\
%&\leq O(pn^{\frc((t-1)r/s)}(1+o(1)))\\
%&\leq O(n^{\frc(tr/s)-1}(1+o(1))).
&=n^{\frc(tr/s)-1}(1\pm o(1)).
\end{align*}

Note that a more careful analysis need also bound the number of vertices participating in $S(t)\cap S(t')$ for all $t'<t$. %, that in the last step, we use the fact that $u, v$ are from disjoint sets $V_{\lfloor tr/s \rfloor}$ and $V_{\lfloor tr/s \rfloor}$ respectively, and hence, the two event $u \in S(t-1)$ and $(u,v) \in E$ are independent. The various vertices in $G$ all have the same probability of membership in $S(t-1)$.
Further, even in this case, %for these independent variables,
 tight concentration assumed above is only achieved when the expected size of the set is $n^{\Omega(1)}$. However, this is guaranteed by the inductive hypothesis, assuming $r$ and $s$ are relatively prime.
\end{proof}

%In particular, we see that the probability of participating in the rightmost vertex set is at most $n^{-1}(1+o(1))$.
Now by symmetry, the same bounds can be given when constructing the candidate sets in the opposite direction, from right to left (note the symmetry of the structure). Thus, %when all leaves are given, the probability of participating in the leftmost vertex set is also $n^{-1}(1+o(1))$. Moreover,
once all the leaves are fixed, every candidate for an internal vertex can be described, for some $t\in[1,s-1]$, as the rightmost backbone vertex in the $t$th prefix, as well as the leftmost backbone vertex in the $(s-t)$th prefix starting from the right. %Conditioning on the sizes of all other sets (and assuming these do not deviate much from their respective expectations, as guaranteed with high probability),
By Claim~\ref{G(n,p)-bound}, the probability
of this event is at most $$n^{\frc(tr/s)-1}n^{\frc((s-t)r/s)-1}(1+o(1))=n^{-1}(1+o(1)).$$
Thus, since the $(r,s)$-caterpillar has $s-r$ internal vertices and $r+1$ leaves, it follows by standard probabilistic arguments that, for some universal constant $C>0$, the probability that
%Therefore, the probability that any of these sets exceeds $\log n$ is at most $n^{-C\log\log n}$ for some universal constant $C>0$, and the probability that the
total number of caterpillars for any  sequence of leaves exceeds $(\log n)^{s-r}$ is at most $(s-r)n^{r+1}n^{-C\log\log n}$, which is $o(1)$ for any constants $r,s$.

Now let us consider the number of (r,s)-caterpillars with a fixed set of leaves in a $k$-subgraph $H$ with minimum degree at least $d=k^{(r+\eps)/s}$. Ignoring low-order terms (which would account for repeated leaves), the number of (r,s)-caterpillars in $H$ (double counting each caterpillar to account for constructing it inductively once from each direction) is at least $kd^s$ (since it is a tree with $s$ edges), whereas the number of possible sequences of leaves is at most $k^{r+1}$. Thus, the  number of $(r,s)$ caterpillars in $H$ corresponding to the average sequence of $r+1$ $H$-vertices is at least $kd^s/k^{r+1}= k^{r+\eps}/k^r=k^\eps$. Note that the parameters for the high probability success of the dense-versus-random distinguishing algorithm are the same as for the random planted model, giving, as before, an distinguishing ratio of $\tilde{O}(n^{1/4})$ in the worst case.

\section{An LP based algorithm for arbitrary graphs}\label{sec:main}

We now give a general algorithm for DkS inspired by the distinguishing algorithm in the Dense vs Random setting. For a graph $G$ with maximum degree $D=n^{r/s}$, we will use the $(r,s)$-caterpillar template, and keep track of sets $S(t)$ as before. We then fix the leaves one by one, while maintaining suitable bounds on $S(t)$.

Let us start by describing the LP relaxation.\footnote{The entire algorithm can be executed without solving an LP, by performing an exhaustive search for the best set of leaves, with running time comparable to that of solving the LP. We will elaborate on this later.}
Taking into account the simplifications from Section~\ref{sec:notation.and.simplification}, we %may assume that our graph has maximum degree at most $D=n^{r/s}$, and
define a hierarchy of LPs which is satisfied by a graph which contains a subgraph of size at most $k$ with minimum degree at least $d$. This hierarchy is at most as strong as the Lov\'asz-Schrijver LP hierarchy based on the usual LP relaxation (and is possibly weaker). Specifically, for all integers $t\geq 1$, we define \textbf{DkS-LP}$_t(G,k,d)$ to be the set of $n$-dimensional vectors $(y_1,\ldots y_n)$ satisfying:
\begin{eqnarray}
 \sum_{i\in V}y_i\leq k, &\;& \text{and}\label{LP:k-bound}\\
 \exists \{y_{ij}\mid i,j\in V\} &\;& \text{s.t.}\nonumber \\
 \forall i\in V &\;& \sum_{j\in\Gamma(i)}y_{ij}\geq dy_i\label{LP:d-bound}\\
 \forall i,j\in V &\;& y_{ij}=y_{ji}\label{LP:symmetry}\\
 \forall i,j\in V &\;& 0\leq y_{ij}\leq y_i\leq 1\label{LP:0-1}\\
 \text{if }t>1\text{, }\forall i\in V\text{s.t. }y_i\neq 0 &\;& \{y_{i1}/y_i,\ldots,y_{in}/y_i\}\in\textbf{DkS-LP}_{t-1}(G,k,d)\label{LP:recursion}
\end{eqnarray}

Given an LP solution $\{y_i\}$, we write LP$_{\{y_i\}}(S) = \sum_{i \in S} y_i$. When the solution is clear from context, we denote the same by LP$(S)$. We call this the \emph{LP-value} of $S$. %Note that the above sequence of relaxations (for increasing values of $t$) is indeed a hierarchy, that is, for all $t'\leq t$ a solution to DkS-LP$_t$ also satisfies DkS-LP$_{t'}$.
When the level in the hierarchy will not be important, we will simply write DkS-LP instead of DkS-LP$_t$. A standard argument shows that a feasible solution to DkS-LP$_t(G,k,d)$ (along with all the recursively defined solutions implied by constraint~\eqref{LP:recursion}) can be found in time $n^{O(t)}$. For completeness, we illustrate this in Appendix~\ref{sec:LP-solution}

Informally, we can think of the LP as giving a distribution over subsets of $V$, with $y_i$ being the probability that $i$ is in a subset. Similarly $y_{ij}$ can be thought of as the probability that both $i,j$ are `picked'. We can now think of the solution $\{y_{ij}/y_i: 1 \leq j \leq n\}$ as a distribution over subsets, conditioned on the event that $i$ is picked.

\paragraph{Algorithm outline.} The execution of the algorithm follows the construction of an $(r,s)$-caterpillar. We perform $s$ steps, throughout maintaining a subset $S(t)$ of the vertices. For each $t$, we perform either a `backbone step' or a `hair step' (which we will describe shortly). In each of these steps, we will either find a dense subgraph, or extend an inductive argument that will give a lower bound on the ratio $\text{LP}(S(t))/|S(t)|$. Finally, we show that if none of the steps finds a dense subgraph, then we reach a contradiction in the form of a violated LP constraint, namely $LP(S(s))>|S(s)|$.

\subsection{The Algorithm}

Let us now describe the algorithm in detail. The algorithm will take two kinds of steps, {\em
backbone} and {\em hair}, corresponding to the two types of caterpillar edges. While these steps differ in the updates they make, both use the same procedure to search locally for a dense subgraph starting with a current candidate-set. Let us now describe this procedure.

\begin{center}\fbox{
\begin{minipage}{14 cm}
\textbf{DkS-Local$(S,k)$}
\begin{itemize}
    \item Consider the bipartite subgraph induced on $(S,\Gamma(S))$.
    \item For all $k'=1,\ldots,k$, do the following:
    \begin{itemize}
        \item Let $T_{k'}$ be the set of $k'$ vertices in $\Gamma(S)$ with the highest degree (into $S$).
        \item Take the $\min\{k',|S|\}$ vertices in $S$ with the most neighbors in $T_{k'}$, and let $H_{k'}(S)$ be the bipartite subgraph induced on this set and $T_{k'}$.
    \end{itemize}
    \item Output the subgraph $H_{k'}(S)$ with the largest average degree.
\end{itemize}
\end{minipage}
}%\fbox
\end{center}

We will analyze separately the performance of this procedure in the context of a leaf-step and that of a hair-step. %This involves a few simple lemmas, all of which have short simple proofs which are deferred to the appendix (Section~\ref{sec:app:proofs}).
We begin by relating the performance of this procedure to an LP solution.

\begin{claim}\label{clm:dkslocal} Given a set of vertices $S\subseteq V$, and an LP solution $\{y_i\}\in \emph{DkS-LP}(G,k,d)$, let $k'=\roundup{\text{\emph{LP}}(\Gamma(S))}$. Then \emph{DkS-Local}$(S,k)$ outputs a subgraph with average degree at least
\[ \frac1{\max\{|S|,k'\}} \cdot \sum_{j\in\Gamma(S)}y_j|\Gamma(j)\cap S|. \]
\end{claim}
%\iffalse %%%%% MOVED TO APPENDIX
\begin{proof} %By constraints~\eqref{LP:0-1} and~\eqref{LP:d-bound} of the LP we have $$\sum_{j\in\Gamma(S)}y_j|\Gamma(j)\cap S|\geq\sum_{j\in\Gamma(S)}\sum_{i\in\Gamma(j)\cap S}y_{ij}=\sum_{i\in S}\sum_{j\in\Gamma(i)}y_{ij}\geq d\text{LP}(S).$$
Note that by constraint~\eqref{LP:k-bound}, $k'=\roundup{\text{LP}(\Gamma(S))}\leq k$. Then in Procedure DkS-Local, the vertices in $T_{k'}$ must have at least $\sum_{j\in\Gamma(S)}y_j|\Gamma(j)\cap S|$ edges to $S$: indeed, the summation $\sum%_{j\in\Gamma(S)}
y_j|\Gamma(j)\cap S|$ can be achieved by taking $\sum_{j\in T_{k'}}1\cdot|\Gamma(j)\cap S|$ and moving some of the weight from vertices in $T_{k'}$ to lower-degree (w.r.t.\ $S$) vertices (and perhaps throwing some of the weight away). After choosing the $\min\{k',|S|\}$ vertices in $S$ with highest degree, the remaining subgraph $H_{k'}(S)$ has average degree at least
\[ \frac{\min\{k',|S|\}}{|S|} \cdot \frac1{k'} \cdot \sum_{j\in\Gamma(S)}y_j|\Gamma(j)\cap S|. \]
This proves the claim.
\end{proof}
%\fi %%%% MOVED TO APPENDIX

The backbone step in the algorithm first performs DkS-Local on the current $S$, and then
sets $S$ to be $\Gamma(S)$. The following lemma gives a way to inductively maintain a lower bound
on LP$(S(t))/|S(t)|$ assuming DkS-Local does not find a sufficiently dense subgraph.

\begin{lemma}\label{lem:expand-guarantee} Given $S\subseteq V$, and an LP solution $\{y_i\}$ for \emph{DkS-LP}$(G,k,d)$: for any $\rho\geq 1$ such that \emph{LP}$(S)/|S|\geq\rho/d$, either \emph{DkS-Local}$(S,k)$ outputs a subgraph with average degree at least $\rho$ %$$\min\left\{\rho,\frac{d\emph{LP}(S)}{|S|}\right\},$$
or we have $$\emph{LP}(\Gamma(S))\geq\frac{d\emph{LP}(S)}\rho.$$
\end{lemma}
%\iffalse %%%%% MOVED TO APPENDIX
\begin{proof}%[Proof (of Lemma~\ref{lem:expand-guarantee})]
By the LP constraints~\eqref{LP:0-1} and~\eqref{LP:d-bound}, we have
\[ \sum_{j\in\Gamma(S)}y_j|\Gamma(j)\cap S|\geq\sum_{j\in\Gamma(S)}\sum_{i\in\Gamma(j)\cap S}y_{ij}=\sum_{i\in S}\sum_{j\in\Gamma(i)}y_{ij}\geq d\text{LP}(S). \]
By Claim~\ref{clm:dkslocal}, Dks-Local$(S,k)$ outputs a subgraph with average degree at least $d\text{LP}(S)/\max\{|S|,k'\}$, where $k'=\text{LP}(\Gamma(S))$ (note that we are ignoring some roundoff error which will be negligible in the context of the algorithm).

If $k'\leq|S|$, then we are done, since by our assumption, $d\text{LP}(S)/|S|\geq\rho$. Now suppose $k'\geq|S|$. The output graph has average degree at least $d\text{LP}(S)/k'$. If this is at least $\rho$, we are done. If not, $k' \geq d\text{LP}(S)/\rho$, and since $k' = \text{LP}(\Gamma(S))$, we get the desired result.
\end{proof}
%\fi %%%% MOVED TO APPENDIX

Let us now consider a hair step. In this case, the algorithm performs DkS-Local on the current
set, and then picks a vertex $j\in V$ to act as a ``leaf". The new $S$ is then set to equal $S\cap\Gamma(j)$.
The following lemmas prove that either DkS-Local finds a sufficiently dense subgraph, or we can pick $j$ so as to
inductively maintain certain bounds. Let us first prove a simple averaging lemma.

%\iffalse %%%%% MOVED TO APPENDIX
\begin{lemma}\label{lem:avging}
Let $x_j$, ($1 \leq j \leq n$) be reals in $[0,1]$, with $\sum_j x_j \leq k$. Let
$P_j$ and $Q_j$ be some non-negative real numbers such that for some $P,Q>0$,
\begin{equation}\label{eqn:avging:hypo}
\sum_j x_j P_j \geq P \mbox{, and } \sum_j x_j Q_j \leq Q.
\end{equation}
Then there exists an $j$ such that $P_j \geq P/(2k)$ and $P_j/Q_j \geq P/(2Q)$.
\end{lemma}
\begin{proof}
By our assumption %$\sum_i x_i \leq k$. Thus
$\sum_j x_j \big( P_j - \frac{P}{2k}
\big) \geq P-\frac{P}{2}=\frac{P}{2}$. Thus from~\eqref{eqn:avging:hypo}, it
follows that there exists an $j$ such that $x_j>0$ and
\[ P_j - \frac{P}{2k} \geq \frac{P}{2Q}\cdot Q_j. \] This choice of $j$ clearly satisfies the required properties.
%For this $i$, there exists an $\alpha$ such that $P_i = \big( \alpha +
%\frac{1}{2} \big) P$, and $Q_i \leq 2Q\alpha$. Thus setting $g = \max\{\alpha,
%\frac{1}{2}\}$ works.
\end{proof}
%\fi %%%% MOVED TO APPENDIX

%The following lemma is crucial in analysing the hair step.

\begin{lemma}\label{lem:contract-guarantee}
Let $S \subseteq V$, and  let $\{y_i\}\in \emph{DkS-LP}(G,k,d)$ be an LP solution (for which there exist corresponding $\{y_{ij}\}$). Then for any $\rho\geq 1$, either
\emph{DkS-Local}$(S,k)$ outputs a $\rho$-dense subgraph, or there exists some vertex $j \in G$, such that $y_j>0$, and
\[ \emph{LP}_{\{y_{ij}/y_j\mid i\in V\}}(S\cap\Gamma(j)) \geq \frac{d
        \cdot \emph{LP}_{\{y_i\}}(S)}{2k}\text{, and}\]
\[  \emph{LP}_{\{y_{ij}/y_j\mid i\in V\}}(S\cap\Gamma(j)) /|S \cap \Gamma(j)| \geq \frac{d\cdot\emph{LP}_{\{y_i\}}(S)}{2\rho \cdot \max\{ k,|S| \}}.\]
\end{lemma}
%\iffalse %%%%% MOVED TO APPENDIX
\begin{proof}
By constraints~\eqref{LP:d-bound} of the LP we have
\begin{equation}\label{eq:contract-guarantee:1}\sum_{j\in\Gamma(S)}y_j \emph{LP}_{\{y_{ij}/y_j\mid i\in V\}}(S\cap\Gamma(j))=\sum_{j\in\Gamma(S)}\sum_{i\in\Gamma(j)\cap S}y_{ij}=\sum_{i\in S}\sum_{j\in\Gamma(i)}y_{ij}\geq d\text{LP}_{\{y_i\}}(S).
\end{equation}
From Claim~\ref{clm:dkslocal}, it follows that if the subgraph found by DkS-Local has average degree less than $\rho$, we must have $\rho>\sum_{j\in\Gamma(S)}y_j|\Gamma(j)\cap S|/\max\{|S|,k'\}$, or in other words
\begin{equation}\label{eq:contract-guarantee:2}\sum_{j\in\Gamma(S)}y_j|\Gamma(j)\cap S|\leq\rho\max\{|S|,k'\}\leq\rho\max\{|S|,k\}.
\end{equation}

Thus the lemma follows immediately from Lemma~\ref{lem:avging} and equations~\eqref{eq:contract-guarantee:1} and~\eqref{eq:contract-guarantee:2}.
\end{proof}
%\fi %%%% MOVED TO APPENDIX

We now formally describe the algorithm. It takes as input a graph $G$, a parameter $k$, and $\{y_i\}$, a solution to DkS-LP$_{r+2}(G,k,d)$. Throughout, a set $S \subseteq V$, and an LP solution $\{ y_i \}$ are maintained. %(@@@is it better to have: compute a solution to LP$_{r+2}(G,k,d)$ as the first step of the algorithm?)

\begin{center}\fbox{
\begin{minipage}{14 cm}
\textbf{DkS-Cat$_{r,s}(G,k,\{y_i\})$}
\begin{itemize}
    \item %Let $\{b_t\}$ be a $(r,s)$-sequence, and choose an initial vertex $j_0$ for which $y_{j_0}>0$.
    Let $S_0=V$.
    \item For all $t=1,\ldots,s$, do the following:
    \begin{itemize}
        \item For $t>1$, let $H_t$ be the output of Procedure DkS-Local$(S_{t-1},k)$.
        \item If the interval $[(t-1)r/s, tr/s]$ contains an integer, perform a {\bf hair step}:\\
        Choose some vertex  $j_t$ as in Lemma~\ref{lem:contract-guarantee} (or for $t=1$, choose any $j_1$ such that $y_{j_1}>0$), and
        \begin{itemize}
            \item Let $S_t=S_{t-1}\cap\Gamma(j_t)$.
            \item Replace the LP solution $\{y_i\}$ with $\{y_{ij_t}/y_{j_t}\mid i\in V\}$.
        \end{itemize}
        \item Otherwise, perform a {\bf backbone step}:\\ Let $S_t=\Gamma(S_{t-1})$.
    \end{itemize}
    \item Output  the subgraph $H_t$ with the highest average degree.
\end{itemize}
\end{minipage}
}%\fbox
\end{center}

Note that since the ``conditioning" (replacing $y_i$'s by $y_{ij}/y_j$) in the hair steps is only performed $r+1$ times, then by constraint~\eqref{LP:recursion}, at every step of the algorithm $\{y_i\}$ satisfies DkS-LP$(G,k,d)$.

\paragraph{A combinatorial algorithm.} Note that the only time the algorithm uses the LP values is in choosing the leaves. Thus, even in the absence of an LP solution, the algorithm can be run by %simply
trying all possible sequences of leaves (the analysis will still work by replacing the LP solution with the optimum $0-1$ solution). While this would take time $O(n^{r+1})$ as opposed to linear time (for the LP-based rounding algorithm), this %running time
is comparable to the time needed to solve the LP. An %somewhat
interesting open question is if it is possible to avoid the %kind of
dependence on $r$, the number of leaves.
%the output of the algorithm does not depend on whether the $(r,s)$-sequence is expanding or contracting, since this only affects the choice of $S_s$ (which is never used, except in the analysis).

\subsection{Performance guarantee}

The analysis is quite straightforward. We follow the various steps, and each time apply either Lemma~\ref{lem:expand-guarantee} or Lemma~\ref{lem:contract-guarantee}, as appropriate. Our main result is the following:

\begin{theorem}\label{thm:main} Let $s>r>0$ be relatively prime integers, let $G$ be an undirected (bipartite) graph with maximum degree $\leq D=n^{r/s}$, let $\{y_i\}\in% be a solution to
\text{\emph{DkS-LP}}_{r+1}(G,k,d)$. % which contains a $k$-subgraph with minimum degree at least $d>0$.
and define $\triv = \max\{Dk/n,1\}$. Then if $d'$ is the average degree of the subgraph found by \emph{DkS-Cat}$_{r,s}(G,k,\{y_i\})$, we have
\[ \max\{d',\triv\}=\Omega(d/D^{(s-r)/s}). \]
\end{theorem}

%The calculations are greatly simplified if we assume $D=n^{r/s}$ (essentially because this will allow us to simplify the second bound in Lemma~\ref{lem:contract-guarantee} when it is applied).

%[@@@ should this come {\em after} the techinical results? or in less detail?]
Note that when the log-density $\alpha$ of $G$ is not rational, we can choose rational $\alpha \leq r/s \leq \alpha +\eps$ for any small $\eps>0$. We then still appeal to Theorem ~\ref{thm:main} as before, though the greedy algorithm might only return a subgraph of average degree $\triv'>\triv/n^{\eps}$. Thus, the loss in the approximation ratio is at most $n^{\eps}$. A fairly straightforward calculation shows that this implies a $O(n^{1/4+\eps})$-approximation in $O(n^{O(1/\eps)})$ time for all $\eps>0$ (including $\eps=1/\log n$).

%As shown earlier, a straightforward calculation shows that Theorem~\ref{thm:main} implies an $O(n^{1/4+\eps})$ approximation in polynomial time.
Before going into the details of the proof, let us note the similarities between the algorithm and the random models discussed earlier. Recall that in the random case, the sets $S(t)$ (corresponding to $S_t$ in the algorithm) had cardinality tightly concentrated around $n^{\frc(tr/s)}$. Similarly here, if we assume that $k=n/D(=D^{(s-r)/r})$, and that $d$ (the density of the subgraph implied by the LP) is at least $\rho k^{r/s}$ (for some $\rho\geq 1$), then we show (see Corollary~\ref{cor:main}) that if until step $t$ the algorithm has not found an $\Omega(\rho)$-dense subgraph, then the current candidate set satisfies
$$\frac{\text{LP}(S_t)}{|S_t|}>\rho^{-\frc(tr/s)\cdot s/r}\left(\frac{k}{n}\right)^{\frc(tr/s)}=\left(\frac{1}{D}\right)^{\frc(tr/s)},$$ which will yield a contradiction after the final step (when $t=s$).

One difficulty is that we avoid making the assumption that $kD=n$ (which is possible, but would incur a $O(\sqrt{\log n})$ loss in the approximation guarantee). Instead, we use the fact that the greedy algorithm finds a $k$-subgraph with average degree $\triv\geq\max\{1,Dk/n\}$. Specifically, we show that at step $t$ of the algorithm, either a subgraph with average degree $\Omega(\rho)$ has already been found, or the greedy algorithm gives the desired approximation (i.e.\ $\triv\geq\rho$), or we have the desired lower bounds on LP$(S_t)$ and LP$(S_t)/|S_t|$.

\noindent {\bf Notation. } In what follows, we %write %$\triv = \max \{Dk/n, 1\}$, and
let $\rho = d/(2D^{(s-r)/s})$ denote the desired average degree of the output subgraph (up to a constant factor). %We will assume that $\rho > \triv~(\geq 1)$, because otherwise a simple greedy algorithm would find a $\rho$-dense subgraph. Further, let us write
We also write $L_t=\rounddown{tr/s}$. Note that the number of hair steps up to and including step $t$ is precisely $L_t+1$.

We now state the main technical lemma. %, which easily implies Theorem~\ref{thm:main}.

\begin{lemma}\label{lem:main} Let $s>r>0$ be relatively prime integers, let $G$ be an undirected (bipartite) graph with maximum degree at most $D=n^{r/s}$, and let $\{y_i\}$ be a solution to \emph{DkS-LP}$_{r+1}(G,k,d)$. % which contains a $k$-subgraph with minimum degree at least $d>0$.
Let $\triv = \max\{Dk/n,1\}$. For $t=1,\ldots,s$, let $d'_t$ be the average degree of the densest of the subgraphs found by \emph{DkS-Cat}$_{r,s}(G,k,\{y_i\})$ up through step $t$. Then either $$\max\{d'_t,\triv\}=\Omega(\rho),$$ or we have
\begin{equation} \label{eq:lem:main}
\begin{split}
\emph{LP}(S_t)&\geq\frac{d^t}{2^{L_t}\rho^{t-L_t-1}k^{L_t}}\text{, and}\\ \frac{\emph{LP}(S_t)}{|S_t|}&\geq\frac{d^t}{2^{L_t}\triv\rho^{t-1}D^{t-L_t}}\text{.}
\end{split}
\end{equation} % where $\rho=d/(2D^{(s-r)/s})$ is the desired average degree of the output subgraph.
\end{lemma}

The following simple corollary %(proved in the appendix, Section~\ref{sec:app:proofs})
 immediately implies Theorem~\ref{thm:main} (by contradiction) when we take $t=s$.
\begin{corollary}\label{cor:main} In the notation of Lemma~\ref{lem:main}, either $\max\{d'_{t},\triv\}=\Omega(d/D^{(s-r)/s})$, or we have
\[ \frac{\emph{LP}(S_t)}{|S_t|}\geq\frac{2^{t-\rounddown{tr/s}}}{D^{\frc(tr/s)}}. \]
\end{corollary}
%\iffalse %%%%% MOVED TO APPENDIX
\begin{proof}
\begin{align*} \frac{\text{LP}(S_t)}{|S_t|}&\geq\frac{d^{t}}{2^{L_t}\triv\rho^{t-1}D^{t-L_t}}&\text{By Lemma~\ref{lem:main}}\\
&>\frac{d^{t}}{2^{L_t}\rho^tD^{t-L_t}}&\rho>\triv\\
&=\frac{d^{t}}{2^{L_t}D^{t-L_t}}\cdot\frac{(2D^{(s-r)/s})^t}{d^t}&\text{definition of }\rho\\
&=\frac {2^{t-L_t}}{D^{t-L_t-t(s-r)/s}}=\frac {2^{t-L_t}}{D^{tr/s-L_t}}.
\end{align*}
\end{proof}
%\fi %%%%% MOVED TO APPENDIX

Let us now proceed to the proof of Lemma~\ref{lem:main}.
\begin{proof}[Proof of Lemma~\ref{lem:main}]
We prove by induction that if the algorithm does not find an $\Omega(\rho)$ dense subgraphs in steps $1$ through $t$, the lower bounds \eqref{eq:lem:main} hold. Assume that $\rho>\triv(\geq 1)$ (otherwise we are done).

For $t=1$, the bounds hold trivially:  If $j_1$ is any vertex for which $y_{j_1}>0$, then we have  LP$_{\{y_{ij_1}/y_{j_1}\}}(S_1)\geq d$ (by constraint~\eqref{LP:d-bound})  and so LP$(S_1)/|S_1|\geq d/|\Gamma(j_1)|\geq d/D(\geq d/(\triv D))$, which is exactly what we need.

Now assume the lemma holds for some $1\leq t\leq s-1$. We will show it for $t+1$, considering separately backbone and hair steps.

First, suppose $t+1$ is a backbone step. Then the interval $[tr/s, (t+1)r/s]$ does not contain an integer, i.e.\ $tr/s<L_t+(s-r)/s$. By Lemma~\ref{lem:expand-guarantee}, if $LP(S_{t})/|S_{t}|\geq\rho/d$ then either the procedure DkS-Local$(S_{t},k)$ produces a  $\rho$-dense subgraph, or we have LP$(S_{t+1})=\text{LP}(\Gamma(S_{t}))\geq d\text{LP}(S_{t})/\rho$. In that case, since $|S_{t+1}|=|\Gamma(S_{t})|\leq D|S_{t}|$, the claim follows immediately from the inductive hypothesis. Thus it suffices to show that indeed $LP(S_{t})/|S_{t}|\geq\rho/d$. This follows from Corollary~\ref{cor:main}, which gives
\[ \frac{\text{LP}(S_{t})}{|S_{t}|}\geq\frac{2^{t-L_{t}}}{D^{tr/s-L_{t}}}>\frac {2^{t-L_{t}}}{D^{(s-r)/s}}=\frac{2^{t+1-L_t}\rho}{d}.\]

Now, suppose $t+1$ is a hair step. Then the interval $[tr/s,(t+1)r/s]$ does contain an integer, i.e.\ $tr/s\geq L_t+(s-r)/s$. Assuming Procedure DkS-Local$(S_t,k)$ does not return a subgraph with average degree at least $\rho$, by Lemma~\ref{lem:contract-guarantee}, there is some choice of vertex $j_{t+1}$ such that

\begin{equation}\label{eq:main.hair.1}
 \text{LP}_{\{y_{ij_{t+1}}/y_{j_{t+1}}\mid i\in V\}}(S_{t+1})=\text{LP}_{\{y_{ij_{t+1}}/y_{j_{t+1}}\mid i\in V\}}(S_t\cap\Gamma(j_{t+1})) \geq \frac{d
        \cdot \text{LP}_{\{y_i\}}(S_t)}{2k}\text{, and}\end{equation}
\begin{equation}\label{eq:main.hair.2}  \text{LP}_{\{y_{ij_{t+1}}/y_{j_{t+1}}\mid i\in V\}}(S_{t+1}) /|S_{t+1}| \geq \frac{d\cdot\text{LP}_{\{y_i\}}(S_t)}{2\rho \cdot \max\{ k,|S_t| \}}.\end{equation}
%
%
%$$|S_{t+1}|=|S_{t+1}\cap\Gamma(j_{t+1})|\leq g\rho|S_t|/\min\{k,|S_t|\}\text,$$ %$$\text{and }\text{LP}_{\{y_{ij_{t+1}}/y_i\mid i\in V\}}(S_{t+1})%=\text{LP}_{\{y_{ij_{t+1}}/y_i\mid %i\in V\}}(S_t\cap\Gamma(j_{t+1}))
%=gd\text{LP}_{\{y_i\mid i\in V\}}(S_t)/(2k)\text.$$
Thus, the required lower bound LP$(S_{t+1})$ follows from~\eqref{eq:main.hair.1} and the inductive hypothesis. If $|S_t|\geq k$, the bound on LP$(S_{t+1})/|S_{t+1}|$ similarly follows from~\eqref{eq:main.hair.2} and the inductive hypothesis. Thus, it remains to show the required bound when  $|S_t| < k$. In that case,~\eqref{eq:main.hair.2} gives %\begin{align*}|S_{t+1}|\leq g\rho\max\{|S_t|/k, 1\}&\leq g\rho\max\{g_t\triv\rho^{B_t}D^{A_t}/k^{B_t+1},1\}\\&=gg_t\triv\rho^{B_{t+1}}D^{A_t}/k^{B_{t+1}}.
%\end{align*}
%Let us show this now.
\begin{align*}
\text{LP}(S_{t+1}) /|S_{t+1}| &\geq \frac{d\text{LP}(S_t)}{2\rho k}\\
&\geq\frac{d^{t+1}}{2^{L_t+1}\rho^{t-L_t}k^{L_t+1}}&\text{inductive hypothesis}\\
&>\frac{d^{t+1}}{2^{L_t+1}\rho^{t}}\cdot\frac{\triv^{L_t}}{k^{L_t+1}}&\text{since $\rho>\triv$}\\
&\geq\frac{d^{t+1}}{2^{L_t+1}\rho^{t}}\cdot\frac{D^{L_t+1}}{\triv n^{L_t+1}}&\text{since $\triv\geq kD/n$}\\
&\geq\frac{d^{t+1}}{2^{L_t+1}\rho^{t}}\cdot\frac{D^{L_t+1}}{\triv n^{(t+1)r/s}}&\text{since $tr/s\geq L_t+(s-r)/s$}\\
&\geq\frac{d^{t+1}}{2^{L_t+1}\rho^{t}}\cdot\frac{1}{\triv D^{t+1-(L_t+1)}}.&\text{since $D\geq n^{r/s}$}\\
%&=\frac{d^{t+1}}{2^{L_{t+1}}\rho^{t}}\cdot\frac{1}{\triv D^{t+1-L_{t+1}}}.&\text{since %$L_{t+1}=L_t+1$}\\
\end{align*}

Since $L_{t+1}=L_t+1$, the bound follows. This concludes the proof.% of the lemma, and hence that of Theorem~\ref{thm:main}.
\end{proof}

\section{An $n^{(1/4 - \eps)}$-approximation with `mildly' exponential run time} \label{sec:subexponential}

We will now consider an extension of our approach to the case when
the log-density of the subgraph $H$ is (slightly) less than the
log-density of the host graph (a crucial case if one wants to go
beyond $n^{1/4}$-approximations). This is done at the expense of
running time -- we obtain a modification of our caterpillar-based
algorithm, which yields an approximation ratio of
$O(n^{(1-\eps)/4})$ approximation which runs in time
$2^{n^{O(\eps)}}$. The main modification is that for each leaf,
rather than picking an individual vertex, we will pick a cluster
of roughly $O(n^{\eps})$ vertices (which is responsible for the
increased running time). The cluster will be used similarly to a
single leaf vertex: rather than intersecting the current set with
the neighborhood of a single vertex, we will intersect the current
set with the neighborhood of the cluster (i.e., the union of all
neighborhoods of vertices in the cluster).

\paragraph{Motivation.}
We now describe briefly why we could expect such a procedure to work
(and why we {\em need} to look at sets of size roughly
$n^{\eps}$). Suppose we are given a random graph $G$ with log-density
$\rho$.  Let $r/s$ denote a rational number roughly equal to $\rho +
\delta$ (for some small constant $\delta$). Suppose we call an $(r+1)$
tuple of leaves `special' if there is an $(r,s)$ caterpillar
`supported' on it, in the sense of Section~\ref{sec:dvsr}. Since $r/s
> \rho$, most $(r+1)$-tuples of vertices are not special.

The crucial observation is the following: say we pick a set $S$ of
vertices, and ask how many $(r+1)$ tuples from $S$ are special. This
number turns out to be `large' (appropriately defined later) roughly
iff $|S| > n^{\delta}$. Now suppose we had planted a random graph $H$
on $k$ vertices and log-density $\rho - \eps$ in $G$. Further suppose
$\delta$ is such that $k^{\eps + \delta} \ll n^{\delta}$ (since $k$ is
much smaller than $n$, we can always choose this so, by setting
$\delta$ to be a constant multiple of $\eps$). By the above claim,
sets in $H$ of size $k^{\eps + \delta}$ would have a `large' number of
special tuples (since $r/s = (\text{log-density of }H) + \eps +
\delta$). But this number, by choice is much smaller than
$n^{\delta}$, thus if there was no $H$ planted, sets of size $k^{\eps
  + \delta}$ would not have many special tuples!

This gives a distinguishing algorithm which runs in time roughly
$\binom{n}{k^{\eps + \delta}}$. Let us now develop the algorithm in
detail, for arbitrary graphs.

To simplify the presentation, we will not describe the LP setting
here, but only consider the equivalent of the combinatorial algorithm
described earlier (``guessing the leaves"). Due to the already large
running time, this will make little difference in terms of
efficiency. Before we introduce the algorithm and provide a sketch of
the analysis, let us state the main result of this section:

\begin{theorem}\label{thm:subexponential} For every $\eps>0$, there is a randomized $O(2^{n^{6\eps}})$-time algorithm which for every graph $G$ with high probability finds a $k$-subgraph whose average degree is within $\tilde{O}(n^{(1-\eps)/4})$ of the optimum.
\end{theorem}

Theorem~\ref{thm:subexponential} is an immediate corollary of the performance guarantee of the modified algorithm, which is as follows:

\begin{theorem}\label{thm:subexponential-technical} For every $0<\eps<\frac12$, and every $0<\beta<1$, there is a randomized algorithm which for every instance of Dense $k$-Subgraph with $k=n^{\beta}$, finds a $k$-subgraph whose average degree is within $\tilde{O}(n^{\beta(1-\beta)(1-\eps)})$ of the optimum (with high probability). Moreover, this algorithm runs in time at most $2^{\tilde{O}(n^{2\beta\eps/(2\beta\eps+(1-\beta))})}$.
\end{theorem}

Let us see that Theorem~\ref{thm:subexponential} follows easily:
\begin{proof}[Proof of Theorem~\ref{thm:subexponential}] For $\eps>1/6$, we can check every subgraph in time $2^n<2^{n^{6\eps}}$. Thus, we may assume that $\eps\leq 1/6$. Let $k=n^\beta$, and consider separately two cases. First, suppose that $\beta\leq5/7$. Note that the algorithm referred to in Theorem~\ref{thm:subexponential-technical} always gives an $\tilde{O}(n^{(1-\eps)/4})$-approximation (regardless of $\beta$). Thus we only need to bound the running time. For $\beta\leq 5/7$, the expression in the exponent is $$n^{2\beta\eps/(2\beta\eps+(1-\beta))}\leq n^{2\beta\eps/(1-\beta)}\leq n^{2(5/7)\eps/(2/7)}=n^{5\eps}.$$

Now consider $\beta>5/7$. Using the reduction in Appendix~\ref{sec:boundproduct}, we may assume that $D=n^\alpha$ for $\alpha<n^{2/7}$. Algorithm DkS-Cat described in Section~\ref{sec:main} (for $r/s\approx\alpha$) already gives (in quasi-polynomial time) up to a constant factor an approximation ratio of at most $$D^{1-\alpha}=n^{\alpha(1-\alpha)}<n^{10/49}<n^{(1-1/6)/4}\leq n^{(1-\eps)/4}.$$
\end{proof}

To simplify the analysis, we may perform additional preprocessing steps which will constrain the range of parameters:

\begin{lemma}\label{lem:subexp-preprocessing} For the purposes of finding an $\tilde{O}(n^{(1-\eps)/4})$-approximation for DkS, we may assume that the maximum degree of $G$ is $D=n/k$, and that moreover, for $\alpha$ s.t.\ $D=n^\alpha$, the minimum degree in $H$ is $d=k^{\alpha(1-\delta)}$ for some $0\leq\delta\leq\eps$.
\end{lemma}
\begin{proof}[Proof (sketch)] Let $\beta$ be s.t.\ $k=n^\beta$. Then perform the following steps:
\begin{enumerate}
  \item If $d>k^{1-\beta}$, prune the edges of $G$ such that every edge is retained independently with probability $\left(k^{1-\beta}/d\right)$.
  \item Now let $D'$ be the $k/2$-largest degree in the pruned graph.
  \begin{itemize}
    \item If $D'k<n$, add to the current set of edges the edges of the random graph $G(n,1/k)$. The new edges cannot make more then a negligible contribution to the densest $k$-subgraph.
    \item Otherwise, if $D'k>n$, prune as in Appendix~\ref{sec:boundproduct}.
  \end{itemize}
\end{enumerate}
  Step 1 already guarantees that $d\leq k^{1-\beta}$, and step 2 guarantees that the $k/2$-largest degree in the final graph is approximately $n/k$. By Lemma~\ref{lem:greedy}, w.l.o.g.\ we can ignore the first $k/2$ vertices. It only remains to check that $d$ is not too small (not less than $k^{(1-\beta)(1-\eps)}$). Arguing as in Appendix~\ref{sec:boundproduct}, we may assume that if $D'k>n$ then in the final graph we have $d>n^{(1-\eps)/4}=k^{(1-\eps)/(4\beta)}\geq k^{(1-\beta)(1-\eps)}$.
\end{proof}

In the modified algorithm, the backbone step and its analysis is the same as before. The main difference, as we mentioned, is in the leaf step. The following lemma (compare to Lemma~\ref{lem:contract-guarantee}) gives the guarantee of a leaf step when the current set is intersected with the neighborhood of a cluster as opposed to the neighborhood of a single vertex.

\begin{lemma}\label{lem:cluster-guarantee}
Let $S \subseteq V$, and let $0<C<k/2$. Then for any $\rho\geq 1$, either
\emph{DkS-Local}$(S,k)$ outputs a $\rho$-dense subgraph, or there exists some cluster of $C$ vertices $J\subset V(H)$, such that either we can extract a $\rho$-dense $k$-subgraph from the bipartite subgraph induced on $(J,S\cap\Gamma(J))$ in linear time, or we have
$$|S\cap\Gamma(J)\cap H| \geq \Omega\left(\frac{d C\cdot |S\cap H|}{\rho k}\right)\text{, and}$$
$$\frac{|S\cap\Gamma(J)\cap H|}{|S \cap \Gamma(J)|} \geq \Omega\left(\frac{d\cdot|S\cap H|}{\rho^2 \cdot \max\{ k,|S| \}}\right).$$
\end{lemma}
%\iffalse %%%%% MOVED TO APPENDIX
\begin{proof}[Proof (sketch)]
The proof follows the same lines as that of Lemma~\ref{lem:contract-guarantee}. This gives similar bounds with respect to $\sum_{j\in J}|\Gamma(j)\cap S\cap H|$. An extra $\rho$ factor is lost since the size of the set $\Gamma(J)\cap S\cap H$ may be significantly less than the number of incoming edges from $J$. However, if it is more than a factor $\rho$ less, then together with $J$ it forms a $\rho$-dense subgraph (and so a matching guarantee can be found by greedily picking highest $J$-degree vertices in S).
\end{proof}

We now present the modified algorithm. It takes as input a graph $G$ with maximum degree at most $n/k$ in which the densest $k$-subgraph has minimum degree in the range $(k^{(1-\beta)(1-\eps)},k^{(1-\beta)})$ where $\beta$ is such that $k=n^\beta$. These assumptions are valid by Lemma~\ref{lem:subexp-preprocessing}.

\begin{center}\fbox{
\begin{minipage}{14 cm}
\textbf{DkS-Exp$_\eps(G,k)$}
\begin{itemize}
    \item Let $\alpha=1-\beta$ (so $D=n^\alpha$). Let $\alpha'=\alpha+2\beta\eps$, and let $0<r<s<\log n$ be integers such that $|r/s-\alpha'|\leq1/\log n$. Let $C=n^{2\beta\eps/(2\beta\eps+\alpha)}$. Let $S_0=V$.
    \item For all $t=1,\ldots,s$, do the following:
    \begin{itemize}
        \item For $t>1$, let $H_t$ be the output of Procedure DkS-Local$(S_{t-1},k)$.
        \item If the interval $[(t-1)r/s, tr/s]$ contains an integer, perform a {\bf hair step}:\\
        ``Guess" a set $J_t$ of $C$ vertices and
        \begin{itemize}
            \item Let $S_t=S_{t-1}\cap\Gamma(J_t)$.
            \item Let $H'_t$ be the output of Procedure DkS-Local$(J_t,k)$ run on the subgraph of $G$ induced on $J_t\cup S_t$.
        \end{itemize}
        \item Otherwise, perform a {\bf backbone step}:\\ Let $S_t=\Gamma(S_{t-1})$.
    \end{itemize}
    \item Output the subgraph $H_t$ or $H'_t$ with the highest average degree.
\end{itemize}
\end{minipage}
}%\fbox
\end{center}

Note that guessing $r$ clusters of size $C$ correctly (such that they
all satisfy Lemma~\ref{lem:cluster-guarantee}) takes time at most ${n
  \choose C}^r=2^{O(C \log n)}$, as required by
Theorem~\ref{thm:subexponential-technical}. The approximation
guarantee is verified by showing that the algorithm finds a subgraph
with density $\rho=\Omega(d/k^{\alpha(1-\eps)})$. The proof runs along
the same lines as Lemma~\ref{lem:main}, using a combination of
Lemma~\ref{lem:expand-guarantee} and Lemma~\ref{lem:cluster-guarantee}
to prove the following lower bounds inductively (assuming no
$\rho$-dense subgraph is found):
$$|S_t\cap H|>\frac{d^tC^{L_t+1}}{\rho^tk^{L_t}}\text{, and}$$
$$\frac{|S_t\cap
  H|}{|S_t|}>\frac{d^t}{\rho^{t+L_t}D^{t-L_t}}\text{.}$$ In
particular, a simple calculation shows that the second bound implies
$$\frac{|S_t\cap H|}{|S_t|}>\frac{1}{D^{tr/s-L_t}}$$ which for $t=s$
is a contradiction (indicating that a $\rho$-dense subgraph must have
been found).

\section{Going beyond log-density in random models}\label{sec:sublogdensity}
We now show that for a certain range of parameters, we can do better than
the log-density bound in Section~\ref{sec:random} and later. Specifically,
we show this for $D<\sqrt{n}$ (by the simplification in Section~\ref{sec:boundproduct},
we can use this interchangeably with the condition $k>\sqrt{n}$, which we
will in this section).

\subsection{Improvements in the Random Planted Model}
In the easiest model, the Random Planted Model, if $D\leq\sqrt{n}$
(that is, $p<1/\sqrt{n}$), then $H$ is simply composed of the $k$
vertices of highest degree in $G$ whenever the expected average
degree in $H$ is $\gg\sqrt{\log n}\sqrt{D}$ (this is because with
high probability, the minimum degree in $H$ is considerably larger
than the standard deviation of degrees in $G$). This is an
improvement over the log-density bound in this range which in the
worst case gives a distinguishing ratio of $D^{1-\alpha}$. A
similar argument using the standard deviation of intersections of
neighborhoods gives an $n^{1/4}$ distinguishing ratio even for
larger $D$, however for $D\geq\sqrt{n}$ it does not surpass the
log-density barrier.

\subsection{Improvements for the Dense in Random problem}
We now show that in the Dense in Random problem, a simple eigenvalue
approach does better than the log-density bound of
Section~\ref{sec:dvsr} for certain ranges of $k$ and the average
degree $D$. Suppose $G$ is a random graph, with log-density $\rho$. In
this case, the results of Section~\ref{sec:dvsr} allow us to
distinguish random $G$ from graphs which have a planted $k$-subgraph
of density $ \gg k^{\rho}$.

Let us now consider the second eigenvalue of the input graph. If it
were a random graph with log-density $\rho$, F\"uredi and
K\'omlos~\cite{furedikomlos} show that $\lambda_2 \leq O(n^{\rho/2})$,
w.h.p. Now suppose $G$ contains a subgraph $H$ on $k$ vertices with
min degree $d \gg n^{\rho/2} + \frac{kn^{\rho}}{n}$. Consider the
vector $x$ (in $\mathbb{R}^n$, with coordinates indexed by vertices of
$G$) defined by
\[ x_i = \begin{cases} 1 \text{, if } i \in V(H)\\ -\frac{k}{n-k} \text{ otherwise} \end{cases} \]
Note that $x$ is orthogonal to the all-ones vector.\footnote{which we
  assume, is the eigenvector corresponding to $\lambda_1$ -- this is
  not quite true, since the graph is not of uniform degree, but it
  turns out this is not a critical issue.} Further, if $A_G$ is the
adjacency matrix of $G$, we have that
\begin{align*}
\frac{x^T A_G x}{x^Tx} &= \frac{\sum_{(i,j) \in E(G)} x_i x_j}{k + \frac{k^2}{n-k}}\\
& \geq \frac{kd - k n^{\rho} \cdot \frac{k}{n-k}}{2k} \geq \frac{d}{2} - \frac{k n^{\rho}}{n}
\end{align*}

Thus if $d \gg n^{\rho/2} + \frac{kn^{\rho}}{n}$, we have $\lambda_2
\gg n^{\rho/2}$, and thus we can distinguish between the two cases. We
can now compare this threshold for $d$ to the one obtained by a
log-density consideration ($k^\rho$). Thus when $k \gg \sqrt{n}$, the
$\lambda_2$-based algorithm does much better (i.e., it can solve the
distinguishing problem for a smaller value of $d$).

\iffalse
It turns out that a natural semidefinite programming relaxation for
DkS (as considered in \cite{FS}) gives precisely the same
`distinguishing factor' as the $\lambda_2$-based approach, for the
Dense vs Random question. An interesting open question is to use this
SDP to give algorithms better than the one of Section~\ref{sec:main}
in the general case, even for some ranges of the parameters.
\fi

\subsection{Improvements for Dense vs Random}
%\label{sdpapp}
For the Dense vs Random question addressed in
section~\ref{sec:dvsr}, we consider a natural SDP relaxation for
the densest $k$ subgraph problem considered by \cite{FS}, and
show that this performs better than the algorithm presented
earlier in the parameter range $k>\sqrt{n}$. The SDP is as follows

\begin{eqnarray}\label{eqn:sdp}
\max \sum_{(i,j) \in E(G)} X_{ij} && \mbox{subject to}\\
\sum_{i} X_{ii} = k \label{p:eq1}\\
\sum_j X_{ij} = kX_{ii} && \mbox{for all $i$}\label{p:eq2} \\
X_{ij} \leq X_{ii} && \mbox{for all $i,j$}\label{p:eq3} \\
X \succeq 0 \label{p:eq4}
\end{eqnarray}
This is a relaxation for the problem and it is easy to see that there exists an SDP solution of value
$|E(H)|$, where $H$ is a $k$-subgraph. However we show that the SDP value for a random graph is upper bounded
by exhibiting a suitable dual solution.
\begin{theorem} \label{thm:sdp}
For a random graph $G(n,p)$ of average degree $D$ ($D=np$), the value of the SDP is at most $k
(\sqrt{D} + k^2 D/n)$ w.h.p.
\end{theorem}
\begin{proof}(of theorem ~\ref{thm:sdp})
Let us consider the dual. We have the variables $t$ (eq. \ref{p:eq1}),
$y_i$ (eq. \ref{p:eq2}), $z_{ij}$ (eq. \ref{p:eq3}). The dual is
\begin{eqnarray*}
\min~ kt && \mbox{subject to }\\
U -A \succeq 0 \\
U_{ii} = t-ky_i-\sum_j z_{ij} && \mbox{for all $i$}\\
U_{ij} = y_i+z_{ij}&& \mbox{for all $i,j$}
\end{eqnarray*}

For a random graph, note that
$$\frac{D}{n}J-A+\lambda_2 I \succeq 0$$ where $J$ is the all-ones
matrix (It is easy to see this considering two cases: the all ones
vector, and vectors orthogonal to it). Hence if we set
$y_i=\frac{D}{n}$, $t=\lambda_2+\frac{kD}{n}$ and the rest to $0$, we
get a feasible dual solution of value $\frac{k^2 D}{n} + k\lambda_2
\leq k^2 \frac{D}{n} + k\sqrt{D}$.

In the last step, we used the fact that the second eigenvalue is at
most $\sqrt{np}$ w.h.p.
\end{proof}

As an immediate corollary, we have
\begin{corollary} \label{cor:sdp}
The SDP (\ref{eqn:sdp}) can be used to distinguish between a $G(n,p)$ random graph with average degree $D=np$ and a graph with a $k$-subgraph
of density $\sqrt{D} + kD/n$.
\end{corollary}

Note that the distinguishing guarantee in Corollary \ref{cor:sdp}
is better than the log-density guarantee from Section
\ref{sec:dvsr} when $k>\sqrt{n}$.

\bibliographystyle{alpha}
\bibliography{DkS}
\appendix

\section{Simplifications}\label{sec:simplifications}

\subsection{Lossless Simplifications}
We present here a few simplifying assumptions which can be made with
at most a constant factor loss in the approximation guarantee. We
begin with a greedy algorithm which allows us to bound the maximum
degree in $G$, taken from~\cite{FKP} (Lemmas~3.2 and~3.3).

\begin{lemma}\label{lem:greedy} There exists a polynomial time algorithm which, given a graph $G$ containing a $k$-subgraph with average degree $d$, outputs a $k$-subgraph $H'$ and an $(n-k/2)$-subgraph $G'$ s.t.
\begin{enumerate}
  \item For some $D>0$, $G'$ has maximum degree $\leq D$ and $H'$ has average degree $\max\{\Omega(Dk/n),1\}$.
  \item Either $G'$ contains a $k$-subgraph with average degree $\Omega(d)$, or $H'$ has average degree $\Omega(d)$.
\end{enumerate}
\end{lemma}
\begin{proof}
Let $U$ be the set of $k/2$ vertices of highest degree in $G$
(breaking ties arbitrarily), and let $D = \min_{u\in U} {\mbox
{deg}}(u)$. Let $U'$ be the set of $k/2$ vertices of $G$ of
highest degree into $U$. Let $H'$ be the graph induced on $U \cup
U'$. Observe that $H'$ has $\Omega(Dk^2/n)$ edges and average degree $\Omega(Dk/n)$
(it also has average degree at least $1$, assuming $G$ has at least $k/2$
edges).

Now let $G'$ be the subgraph induced on $V\setminus U$ (note that by definition, $G'$ has maximum degree $\leq D$). % (and all its incident %edges) from
%$G$, and apply all other algorithms for finding the dense
%$k$-subgraph on the graph that remains, and has maximum degree
%$D$.
Let $H$ be the densest $k$-subgraph in $G$, and let $\alpha$ be the fraction of edges of $H$
that are incident with vertices of $U$. If $\alpha \le 1/2$, then
the $dk/2$ of the edges in $H$ remain in $G'$ (and thus $G'$ still has a $k$-subgraph with average degree $\Omega(d)$.
%value of the optimal solution in the graph $G \setminus U$ is
%roughly the same original optimal solution, and hence no
%(significant) harm was done by removing $U$.
On the other hand, if $\alpha \ge 1/2$
then it is not hard to see that the $H'$ must have average degree $\Omega(d)$. %will in fact have average degree which is at most a constant
%factor smaller than the optimal solution, and the approximation
%ratio becomes a constant.
\end{proof}

%\subsection{$D$ bounds the maximum degree of $G$}\label{sec:kD/n}

\iffalse
It would be convenient for us to view $G$ as a graph with maximum
degree $D$ rather than average degree $D$. For random graphs,
maximum and average degrees are roughly the same. But for general
graphs, they are not. In our context, it will be possible to think
of $D$ as the maximum degree, without losing any benefits that one
may obtain by having $D$ also be the average degree. The only use
that we have for $D$ being an average degree rather than maximum
degree is that we can claim that $G$ necessarily has a
$k$-subgraph with average degree $Dk/n$ (a random subgraph will do
in expectation). The following argument (taken from~\cite{FKP})
allows us to make a related claim based on maximum degree rather
than average degree.

In summary, we may assume that $G$ has maximum degree $D$, and
that nevertheless, one can extract a $k$-subgraph of average
degree at least $\max\{Dk/n,1\}$.

\subsection{It suffices to treat $k$ as an upper bound}
\label{sec:lek}
\fi
We now justify the weakened requirement, that the algorithm should return a subgraph of size \emph{at most} $k$ (rather than exactly $k$).

\begin{lemma}\label{lem:small_k} Given an algorithm which, whenever $G$ contains a $k$-subgraph of average degree $\Omega(d)$ returns a $k'$-subgraph of average degree $\Omega(d')$, for some (non specific) $k'<k$), we can also find a (exactly) $k$-subgraph with average degree $\Omega(d')$ for such $G$.
\end{lemma}
\begin{proof}  Apply the algorithm repeatedly, each time removing from $G$ the edges of the
subgraph found. Continue until the union of all subgraphs found
contains at least $k$ vertices. The union of subgraphs each of average degree $d'$ has average
degree at least $d'$. Hence there is no loss in approximation ratio if we
reach $k$ vertices by taking unions of smaller graphs. Either we have not removed half the edges from the
optimal solution (and then all the subgraphs found -- and hence their union -- have average degree $\Omega(d')$), or we have removed half the edges of the optimal solution, in which case our subgraph has at least $dk/2$ edges.

Note that this algorithm may slightly overshoot (giving a subgraph on up to $2k$ vertices), in which case we can greedily prune the lowest degree vertices to get back a $k$-subgraph with the same average degree (up to a constant factor).
\end{proof}

\iffalse
We seek to find a subgraph on $k$ vertices with highest average
degree. However, our algorithms will not necessarily output
subgraphs induced on $k$ vertices. They will in general output
subgraphs with ``high" average degree induced on $k' \not= k$ vertices. There will be two cases
to consider.

 When $k' < k$.

Also, sometimes we return a subgraph of size $2k$ instead of $k$, but this only leads
to a loss of a constant factor by sampling $k$ vertices from it.

\subsection{$G$ may be assumed to be bipartite}
\label{sec:bipartite}

\fi

We now treat the assumption that $G$ (and hence $H$) is bipartite (as in the rest of the paper, we define the approximation ratio of an algorithm $A$ as the maximum value attained by $|\text{OPT}(G)|/|A(G)|$):

\begin{lemma}\label{lem:bipartite} Given an $f(n)$-approximation for \dks\ on $n$-vertex bipartite graphs, we can approximate \dks\ on arbitrary graphs within a $\Omega(f(2n))$-factor.
\end{lemma}
\begin{proof} Take two copies of the vertices of $G$, and connect copies of vertices (in two different sides) which are connected in $G$. Thus the densest $2k$-subgraph in the new bipartite graph has at least the same average degree as the densest $k$-subgraph in $G$. Now take the subgraph found by the bipartite \dks\ approximation on the new graph, and collapse the two sides (this cannot reduce the degree of any vertex in the subgraph). Note that the subgraph found may be a constant factor larger than $k$, in which case, as before, we can greedily prune vertices and lose a constant factor in the average degree.
\end{proof}

\iffalse
A simple randomized reduction from the general case to the
bipartite case allows us to assume the $G$ is bipartite with just
a constant factor loss. Take a random cut of $G$ and remove all other
edges. $G$ becomes bipartite, and in expectation half the edges of
$H$ remain. This reduction can be derandomized by standard
techniques (choose to which side a vertex belongs at random in a
pairwise independent way).

\subsection{Bounding the minimum degree of $H$}
\label{sec:ged}
\fi

Finally, we note that it suffices to consider minimum degree in $H$,
rather than average degree.

\begin{lemma}\label{min_deg_in_H} Approximation algorithms for \dks\ lose at most a constant factor by assuming the densest $k$-subgraph has minimum degree $d$ as opposed to average degree.
\end{lemma}
\begin{proof}
Let $H$ be the densest $k$-subgraph in $G$. %
%It will be convenient for us to view $H$ (the densest $k$-subgraph
%of $G$) as having minimum degree $d$ rather than average degree
%$d$. There is a standard argument for converting average degree to
%minimum degree, and this is by
Iteratively remove from $H$ every
vertex whose remaining degree into $H$ is less than $d/2$. Since
less than $kd/2$ edges are removed by this procedure, a subgraph
on $k' \le k$ vertices and minimum degree at least $d/2$ must
remain. As noted in Lemma~\ref{lem:small_k}, the fact that now we
will be searching for subgraphs with $k' \not= k$ vertices can affect
the approximation guarantee by at most a constant factor. %Also, the bound $kD \le n$ is preserved (in case
%we have made such an assumption to begin with).
\end{proof}

\subsection{Bounding the product $kD$}
\label{sec:boundproduct}

It may simplify some of our manipulations if we can assume that
$kD \le n$. While we do not formally make this assumption anywhere,
it is often instructive to consider this case for illustrative purposes.
In fact, we can make this assumption, however at the cost of introducing randomness to the algorithm, and, more significantly, losing an $O(\sqrt{\log n})$ factor in the approximation guarantee.
%Specifically, we will make this assumption in the algorithm
%of Section~\ref{sec:2/7} (which is given for illustration purposes only).
%In the proof of the main result, Theorem~\ref{thm:main}, we will not make
%this assumption.

%Here we provide an argument why this assumption can be
%made. However, this is randomized and we lose a $O(\sqrt{\log n})$ factor %in the process.

Assume that contrary to our assumption, $D > n/k$. In this case,
take a random subgraph $G'$ of $G$ by keeping every edge with
probability $n/(kD)$. The new maximum degree $D'$ now satisfies
roughly $D' = Dn/(kD) = n/k$ as desired. A standard probabilistic
argument (involving the combination of a Chernoff bound and a
union bound) shows that for every vertex induced subgraph, if $d$
denotes its average degree in $G$, then its average degree in $G'$
is at most $O((1 + dD'/D)(1 + \sqrt{\log n/(1 + dD'/D)}))$.

Let $\rho$ be the approximation ratio that we are aiming at.
For $H$, the densest $k$-subgraph, we may assume that $d
\ge \rho \sqrt{\log n} kD/n$, as otherwise the greedy algorithm in
Lemma~\ref{lem:greedy}
provides a $\rho \sqrt{\log n}$ factor approximation. (This is
where our argument loses more than a constant factor in the
approximation ratio). Hence the average degree
of $H$ in $G'$ is at least $\rho \sqrt{\log n}$. An approximation
algorithm with ratio $\rho$ will then find in $G'$ a subgraph of
average degree at least $\sqrt{\log n}$. Such a subgraph must have
average degree a factor of $D/D'$ higher in $G$, giving an
approximation ratio of $\rho$.

\section{Solving the LP}\label{sec:LP-solution}

We will show that the linear program DkS-LP$_t(G,k,d)$ can be solved in time $n^{O(t)}$ by giving an equivalent formulation with explicit constraints which has total size (number of variables and constraints) $n^{O(t)}$. We start by defining a homogenized version of the basic LP. We say that a vector $(y_1,\ldots,y_n,h)$ is a solution to \textbf{DkS-hom}$_t(G,k,d)$ if it satisfies:

\begin{eqnarray*}
 \sum_{i\in V}y_i\leq kh, &\;& \text{and}\\
 h\geq 0\\
 \exists \{y_{ij}\mid i,j\in V\} &\;& \text{s.t.}\\
 \forall i\in V &\;& \sum_{j\in\Gamma(i)}y_{ij}\geq dy_i\\
 \forall i,j\in V &\;& y_{ij}=y_{ji}\\
 \forall i,j\in V &\;& 0\leq y_{ij}\leq y_i\leq h\\
 \text{if }t>1\text{, }\forall i\in V &\;& \{y_{i1},\ldots,y_{in},y_i\}\in\textbf{DkS-hom}_{t-1}(G,k,d)
\end{eqnarray*}

Note now that the condition $(y_1,\ldots,y_n)\in\text{DkS-LP}_t(G,k,d)$ is now equivalent to $(y_1,\ldots,y_n,1)\in\text{DkS-hom}_t(G,k,d)$. Moreover, we can flatten out the recursive definition of DkS-hom$_t$ to give a fully explicit LP on $n^{t+O(1)}$ variables and $n^{t+O(1)}$ constraints, which can be solved in time $n^{O(t)}$ by standard techniques.

\end{document}